\pdfoutput=1

\documentclass[10pt, a4paper]{article}

\usepackage[utf8]{inputenc}
\usepackage[T1]{fontenc}
\usepackage{lmodern} 

\usepackage[margin=2cm]{geometry}

\usepackage[activate={true,nocompatibility},final,tracking=true,
kerning=true,spacing=true,factor=1100,stretch=10,shrink=10]{microtype}
\SetTracking{encoding={*}, shape=sc}{40}

\usepackage{ellipsis}         
\usepackage[abbreviations,british]{foreign} 

\usepackage{xcolor}
\usepackage{hyperref}
\usepackage[xcolor, hyperref, notion, paper]{knowledge}

\definecolor{darkmidnightblue}{rgb}{0.0, 0.2, 0.4}
\definecolor{persianplum}{rgb}{0.44, 0.11, 0.11}

\hypersetup{
  colorlinks  = true,        
  citecolor   = persianplum, 
  urlcolor    = persianplum, 
  linkcolor   = persianplum
}

\usepackage{amssymb} 
\usepackage{amsmath}
\usepackage{amsthm}

\usepackage[capitalise, noabbrev]{cleveref}

\newtheorem{thm}{Theorem}[section]
\newtheorem{fact}[thm]{Fact}
\newtheorem{lemma}[thm]{Lemma}
\newtheorem{proposition}[thm]{Proposition}

\usepackage{caption}

\usepackage{tikz}
\usetikzlibrary{hobby,backgrounds,calc,trees}
\usetikzlibrary{arrows, decorations.markings, shapes, calc}
\usepackage{float}

\newcommand{\orcid}[1]{\href{https://orcid.org/#1}{\includegraphics[width=9pt]{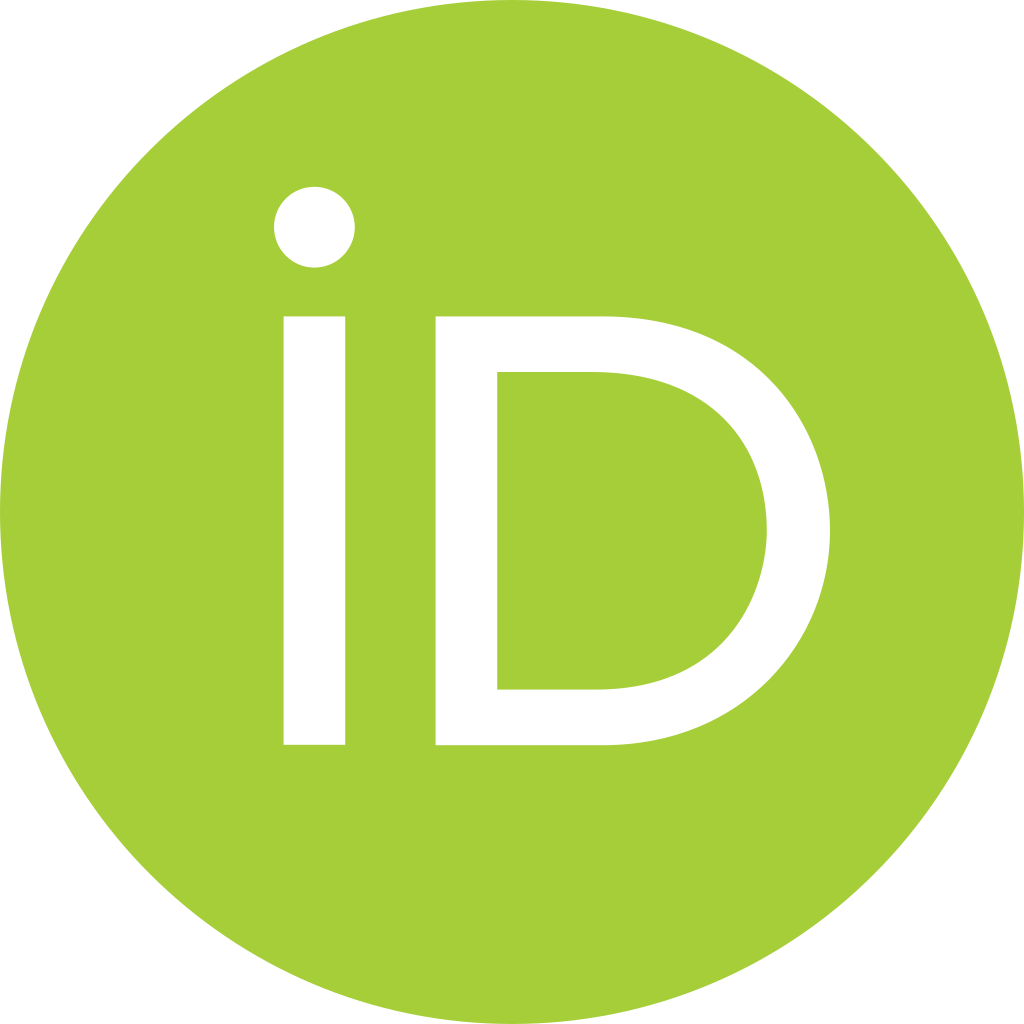}}}

\usepackage{enumerate}

\newrobustcmd{\newnotion}[1]{\AP\intro{#1}}

\newcommand{\dland}{\sqcap}
\newcommand{\topconcept}{\top}
\newcommand{\dlsubseteq}{\sqsubseteq}

\newcommand{\DL}[1]{\ensuremath{\mathcal{#1}}}                          
\newcommand{\ALC}{\DL{ALC}}                                             
\newcommand{\EL}{\kl[EL concept descriptions]{\DL{EL}}}                 
\newcommand{\ELneg}{\kl[ELneg concept descriptions]{\DL{EL}^{(\neg)}}}  
\newcommand{\SEL}{\kl[SEL]{\DL{SEL}}}                                   

\newcommand{\complexityclass}[1]{\textsc{#1}}               
\newcommand{\ExpTime}{\complexityclass{ExpTime}}            
\newcommand{\lang}[1]{\mathbf{#1}}                          
\newcommand{\Rlang}{\kl[role names]{\lang{N_R}}}            
\newcommand{\Clang}{\kl[concept names]{\lang{N_C}}}         

\newcommand{\role}[1]{\kl[role]{\mathit{#1}}}       
\newcommand{\roler}{\role{r}}                       

\newcommand{\concepts}{\lang{C}}                                                    
\newcommand{\elconcepts}{\kl[EL concept descriptions]{\concepts_{\EL}}}             
\newcommand{\elnegconcepts}{\kl[ELneg concept descriptions]{\concepts_{\ELneg}}}    

\newcommand{\concept}[1]{\kl[atomic concept]{\mathrm{#1}}}       

\newcommand{\conceptA}{\concept{A}}                
\newcommand{\conceptB}{\concept{B}}                
\newcommand{\conceptC}{\concept{C}}                
\newcommand{\conceptD}{\concept{D}}                
\newcommand{\conceptE}{\concept{E}}                
\newcommand{\conceptnegA}{\bar{\conceptA}}         
\newcommand{\conceptnegB}{\bar{\conceptB}}         

\newcommand{\conceptReal}{\concept{Real}}                

\newcommand{\plus}[1]{{#1}_{+}}
\newcommand{\conceptAplus}{\plus{\conceptA}}           
\newcommand{\conceptBplus}{\plus{\conceptB}}           
\newcommand{\conceptRealplus}{\plus{\conceptReal}}     

\newcommand{\minus}[1]{{#1}_{-}}
\newcommand{\conceptAminus}{\minus{\conceptA}}           
\newcommand{\conceptBminus}{\minus{\conceptB}}           
\newcommand{\conceptRealminus}{\minus{\conceptReal}}     

\newcommand{\ontology}[1]{\mathcal{#1}}                   
\newcommand{\ontologyO}{\kl[ontology]{\ontology{O}}}      
\newcommand{\ontologyOhard}{\kl[ontologyOhard]{\ontologyO}} 
\newcommand{\ontologyOred}{\kl[ontologyOred]{\ontologyO_{\textit{red}}}} 
\newcommand{\ontologyOcorr}{\kl[ontologyOcorr]{\ontologyO_{\textit{corr}}}} 
\newcommand{\ontologyOtr}{\kl[ontologyOtr]{\ontologyO_{\textit{tr}}}} 

\newcommand{\inter}[1]{\kl[interpretation]{\mathcal{#1}}}    
\newcommand{\interI}{\inter{I}}                              
\newcommand{\interJ}{\inter{J}}                              

\newcommand{\DeltaInter}[1]{\kl[domain]{\Delta}^{#1}}       
\newcommand{\DeltaI}{\DeltaInter{\interI}}                  
\newcommand{\DeltaJ}{\DeltaInter{\interJ}}                  

\newcommand{\cdotInter}[1]{\kl[interpretation function]{\cdot}^{#1}}    
\newcommand{\cdotI}{\cdotInter{\interI}}                                

\newcommand{\domelem}[1]{\mathrm{#1}}                           
\newcommand{\domelemd}{\domelem{d}}                             
\newcommand{\domeleme}{\domelem{e}}                             

\newcommand{\deff}{:=}

\newcommand{\Q}{\mathbb{Q}}

\newcommand{\mymodels}{\kl[satisfies]{\;\models{}}}
\newcommand{\pconditional}[4]{\left( #1 \mid #2\right)\left[ #3, #4 \right]}

\newcommand{\ohardconcepts}{\kl[Ohardconcepts]{\concepts_{\ontologyOhard}}} 

\knowledge{role names}[role name]{notion, color=black}
\knowledge{concept names}[concept name]{notion, color=black}
\knowledge{atomic concept}[concept | Concepts | concepts]{notion, color=black}
\knowledge{top concept}{notion, color=black}
\knowledge{role}[roles]{notion, color=black}
\knowledge{size}[sizes]{notion, color=black}
\knowledge{interpretation}[interpretations]{notion, color=black}
\knowledge{interpretation function}{notion, color=black}
\knowledge{domain}[domains]{notion, color=black}
\knowledge{general concept inclusion}[GCI | GCIs]{notion, color=black}
\knowledge{ontology}[ontologies]{notion, color=black}
\knowledge{satisfies}[satisfaction | semantics | model]{notion, color=black}
\knowledge{consistent}[satisfiable | inconsistent | unsatisfiable | consistency]{notion, color=black}
\knowledge{consistency problem}[consistency problems]{notion, color=black}

\knowledge{EL concept descriptions}[\( \EL \) concept descriptions | \( \EL \) concepts | \( \EL \) concept | \( \EL \)]{notion, color=black}
\knowledge{ELneg concept descriptions}[\( \ELneg \) concept descriptions | \( \ELneg \) concepts | \( \ELneg \) concept | \( \ELneg \)]{notion, color=black}
\knowledge{normal form}{notion, color=black}

\knowledge{SEL}[\( \SEL \)]{notion, color=black}
\knowledge{probabilistic conditional}[probabilistic conditionals]{notion, color=black}

\knowledge{ontologyOhard}[\( \ontologyOhard \)]{notion, color=black}
\knowledge{ontologyOred}[\( \ontologyOred \)]{notion, color=black}
\knowledge{ontologyOcorr}[\( \ontologyOcorr \)]{notion, color=black}
\knowledge{ontologyOtr}[\( \ontologyOtr \)]{notion, color=black}
\knowledge{Ohardconcepts}[\( \ohardconcepts \)]{notion, color=black}

\title{Statistical \texorpdfstring{\(\EL{}\)}{EL} is \texorpdfstring{\(\ExpTime{}\)}{ExpTime}-complete}

\author{Bartosz Bednarczyk \orcid{0000-0002-8267-7554}{}}
\date{Computational Logic Group, Technische Universit{\"a}t Dresden, Germany\\
Institute of Computer Science, University of Wroc{\l}aw, Poland}

\begin{document}
\maketitle

\begin{abstract}
We show that the \kl{consistency problem} for Statistical \( \EL \) \kl{ontologies}, defined by Pe{\~{n}}aloza and Potyka, is \( \ExpTime \)-hard.
Together with existing  \( \ExpTime \) upper bounds, we conclude \( \ExpTime \)-completeness of the logic.
Our proof goes via a reduction from the \kl{consistency problem} for~\( \EL \) extended with negation of atomic concepts.
\end{abstract}

\section{Introduction}\label{sec:intro}

  Description logics (DLs)~\cite{dlbook} are a prominent family of logical formalisms tailored to knowledge representation.
  Nowadays, real-world problems require the ability to handle uncertain knowledge. 
  To deal with this issue, several probabilistic extensions of description logics were proposed in the past~\cite{CarvalhoLC17,Lukasiewicz08,Gutierrez-Basulto17,PenalozaP17}
  Among such extensions, the authors of~\cite{PenalozaP17} proposed Statistical \( \EL \), a statistical variant of the well-known description logic \( \EL \)~\cite{BaaderBL05} famous for tractability of most of its reasoning task.

  In this note we establish tight complexity bounds for the consistency problem for statistical \( \EL \), closing the complexity gaps from~\cite{PenalozaP17}.
  We show that in sharp contrast to its non-probabilistic version, Statistical \( \EL \) is \( \ExpTime \)-complete and hence, provably intractable.
  The main novelty here is the \( \ExpTime \) lower bound, while the \( \ExpTime \) upper bound follows from recent work by Baader and Ecke~\cite[Corollary 15]{BaaderE17} or, alternatively, from work on probabilistic \( \ALC \) by Lutz and Schr{\"{o}}der~\cite[Theorem 9]{LutzS10}.

\section{Preliminaries}\label{sec:prelim}

  In this section, we recall the basics on description logics (DLs) \( \EL \) and \( \ELneg \).
  For readers unfamiliar with DLs we recommend consulting the textbook~\cite{dlbook}, especially Chapters 2.1--2.3, 5.1 and 6.1.

  We fix countably-infinite disjoint sets of \newnotion{concept names} \( \Clang \) and \newnotion{role names} \( \Rlang \).
  Starting from \( \Clang \) and~\( \Rlang \), the set \( \elconcepts \) of \newnotion{\( \EL \) concept descriptions} (or simply \kl{\( \EL \) concepts})~\cite{BaaderBL05} is built using \emph{conjunction} \((\conceptC \dland \conceptD) \), \emph{existential restriction} (\(\exists{\roler}.\conceptC \)) and the \emph{top concept} (\(\topconcept \)), with the grammar~below:
  \begin{equation*} \label{eq:el-grammar}
  \conceptC, \conceptD \; ::= \; \topconcept \; \mid \; \conceptA \; \mid \; \conceptC \dland \conceptD \; \mid \; \exists{\roler}.\conceptC,
  \end{equation*}
  where \(\conceptC,\conceptD \in \elconcepts \), \(\conceptA \in \Clang \) and \(\roler \in \Rlang \). 
  An \(\EL \) \newnotion{general concept inclusion} (\kl{GCI}) has the form \(\conceptC \dlsubseteq \conceptD \) for \kl{\( \EL \) concepts} \(\conceptC, \conceptD \in \elconcepts \). 
  An \(\EL \) \newnotion{ontology} is a finite non-empty set of \(\EL \) \kl{GCIs}.  
  The \newnotion{size} of an \( \EL \) \kl{ontology} is the total number of \(\top\), \kl{role names}, \kl{concept names} and connectives occurring in it.

  \begin{table}[!htb]
        \caption{\kl{Concepts} and \kl{roles} in \(\EL \).\label{tab:EL}}
        \centering
          \begin{tabular}{@{}l@{\ \ \ }c@{\ \ \ }l@{}}
              \hline\\[-2ex]
              Name & Syntax & Semantics \\ \hline \\[-2ex]
              top & \(\topconcept \) & \(\DeltaI  \) \\
              atomic concept & \( \conceptA \) & \(\conceptA^{\interI} \subseteq \DeltaI  \) \\ 
              role & \(\roler \) & \(\roler^{\interI} \subseteq \DeltaI {\times} \DeltaI \) \\ 
              concept\ intersection & \(\conceptC \dland \conceptD \)& \(\conceptC^{\interI}\cap \conceptD^{\interI} \) \\  
              existential\ restriction & \(\exists{\roler}.\conceptC \) & 
              \( \big\{ \domelemd \; | \; \exists{\domeleme}.(\domelemd,\domeleme)\in \roler^{\interI} \land \domeleme\in \conceptC^{\interI} \big\} \)
      \end{tabular}
  \end{table}

  The \kl{semantics} of \( \EL \) is defined via \newnotion{interpretations} \(\interI = (\DeltaI, \cdotI) \) composed of a \emph{finite} non-empty set \(\DeltaI \) called the \newnotion{domain} of \(\interI \) and an \newnotion{interpretation function} \(\cdotI \) mapping \kl{concept names} to subsets of \(\DeltaI \), and \kl{role names} to subsets of \(\DeltaI \times \DeltaI \). 
  This mapping is extended to \newnotion{concepts}, \newnotion{roles} (\cf{}~\cref{tab:EL}) and finally used to define \kl{satisfaction} of \kl{GCIs}, namely \( \interI \mymodels \conceptC \dlsubseteq \conceptD \) iff \(\conceptC^{\interI} \subseteq \conceptD^{\interI}  \).
  We say that an \kl{interpretation}~\(\interI \) \newnotion{satisfies} an \kl{ontology} \( \ontologyO \) (or \(\interI \) is a \kl{model} of \( \ontologyO \), written: \(\interI \mymodels \ontologyO \)) if it \kl{satisfies} all \kl{GCIs} from \( \ontologyO \). 
  An \kl{ontology} is \newnotion{consistent} if it has a \kl{model} and \kl{inconsistent} otherwise. 
  In the \newnotion{consistency problem} for \( \EL \) we ask if an input \( \EL \) \kl{ontology} is \kl{consistent}.
  Note that the \kl{consistency problem} for \( \EL \) is trivial, \ie{} every \( \EL \) \kl{ontology} is \kl{consistent}.

  \subsection{\texorpdfstring{\(\EL{}\)}{EL} with atomic negation}

    The next definitions concern \( \ELneg \), the extension of \( \EL \) with negation of atomic \kl{concepts}.
    More precisely, the set \( \elnegconcepts \) of \newnotion{\( \ELneg \) concepts} is defined by a slight extension of the BNF grammar for \( \EL \):
    \begin{equation*} \label{eq:elneg-grammar}
    \conceptC, \conceptD \; ::= \; \topconcept \; \mid \; \conceptA \; \mid \; \conceptnegA \; \mid \; \conceptC \dland \conceptD \; \mid \; \exists{\roler}.\conceptC,
    \end{equation*}
    where \(\conceptC,\conceptD \in \elnegconcepts \), \(\conceptA \in \Clang \) and \(\roler \in \Rlang \). 
    The \kl{semantics} of \( \ELneg \) \kl{concepts} is defined as in~\cref{tab:EL} with the exception that the \kl{concepts} of the form \( \conceptnegA \) have the \kl{semantics} \( \conceptnegA^{\interI} = \DeltaI \setminus \conceptA^{\interI} \). 
    The notions of \kl{GCIs}, \kl{ontologies} and the \kl{consistency problem} are lifted to \( \ELneg \) in an obvious way.
    We stress that in the presence of negation the \kl{consistency problem} for \( \ELneg \) is no longer trivial and actually is \( \ExpTime \)-complete~\cite[Theorem~6]{BaaderBL05}.\footnote{
    In the setting of \cite{BaaderBL05} the \kl{domains} of \kl{interpretations} might have unrestricted (\ie{} not necessarily finite) sizes, but the result is also applicable to our scenario since \( \ELneg \) has finite model property, which follows from~\cite[Corr.~3.17]{dlbook}.}
    \begin{proposition}\label{prop:exp-hard}
    The \kl{consistency problem} for \( \ELneg \) \kl{ontologies} is \( \ExpTime \)-hard.
    \end{proposition}
   
  \subsection{Statistical \texorpdfstring{\(\EL{}\)}{EL}}

    Statistical \( \EL \), abbreviated as \newnotion{\( \SEL \)}, is a probabilistic DL introduced recently by Pe{\~{n}}aloza and Potyka~\cite[Section~4]{PenalozaP17} to reason about statistical properties over finite  \kl{domains}.
    Statistical \( \EL \) \kl{ontologies} are composed of \newnotion{probabilistic conditionals} of the form \( \pconditional{\conceptC}{\conceptD}{k}{l} \), where \( \conceptC, \conceptD \) are \( \EL \) \kl{concepts} from \( \elconcepts \) and \( k,l \in \Q \) are rational numbers satisfying \( 0 \leq k \leq l \leq 1 \).
    The \kl{size} of \( \SEL \) \kl{ontologies} is defined as in \( \EL \) except that the numbers in \kl{probabilistic conditionals} also contribute to the size and are measured in binary. 

    We say that an \kl{interpretation} \( \interI \) \kl{satisfies} a \kl{probabilistic conditional} \( \pconditional{\conceptC}{\conceptD}{k}{l} \) if:
    \[
      \text{either} \; \conceptD^{\interI} = \emptyset \; \text{or} \; k \leq \frac{| (\conceptC \dland \conceptD)^{\interI} |}{|\conceptD^{\interI}|} \leq l.
    \]
    Note that usual \( \EL \) \kl{GCIs} \( \conceptC \dlsubseteq \conceptD \) are equivalent to \( \pconditional{\conceptD}{\conceptC}{1}{1} \) (\cf~\cite[Proposition~4]{PenalozaP17}). 
    Hence, each \( \EL \) \kl{ontology} can be seen as an \( \SEL \) \kl{ontology} and we can freely use \kl{GCIs} in place of \kl{probabilistic conditionals}.

\section{Main result}\label{sec:main_result}

After introducing all the required definitions, we are ready to prove the main result of this note, namely:
\begin{thm}\label{thm:main}
The \kl{consistency problem} for \( \SEL \) is \( \ExpTime \)-complete. 
\end{thm}

The \( \ExpTime \) upper bound follows from~\cite[Corollary 15]{BaaderE17} or~\cite[Theorem 9]{LutzS10}, hence we focus on the lower bound only.
Let \newnotion{\( \ontologyOhard \)} be an arbitrary \( \ELneg \) ontology. 
With~\newnotion{\( \ohardconcepts \)} we denote the set of all \kl{concept names} that appear (possibly under negation) in \( \ontologyOhard \).
We next design an \( \SEL \) \kl{ontology} \( \ontologyOred \) such that \( \ontologyOred \) is \kl{consistent} iff \( \ontologyOhard \) is and that \( \ontologyOred \) is only polynomially larger than~\( \ontologyOhard \).
It~will be composed of two \( \SEL \) \kl{ontologies}, \( \ontologyOtr \) and \( \ontologyOcorr \), responsible respectively for ``translating''~\( \ontologyOhard \) into~\( \SEL \) and for guaranteeing the correctness of the translation.

The main idea of the encoding is as follows. 
We first produce for each \kl{concept name} \( \conceptA \) from \( \ohardconcepts \)  two fresh, different, \kl{concepts} \( \conceptAplus, \conceptAminus \not\in \ohardconcepts \) intuitively intended to contain, respectively, all members of \( \conceptA \) and from its complement.
Due to the lack of negation, we clearly are not able to fully formalise the above intuition, but the best we can do is to enforce, with the \kl{ontology} \( \ontologyOcorr \), that these \kl{concepts} are interpreted as disjoint sets and each of them contains exactly half of the \kl{domain}.
This is sufficient for our purposes, since with fresh, pair-wise different, \kl{concepts} \(  \conceptReal, \conceptRealplus, \conceptRealminus \not\in \ohardconcepts \) we can separate the ``real'' \kl{model} of \( \ontologyOhard \) from the auxiliary parts required for the encoding. 
Finally, in the ``translation'' \kl{ontology} \( \ontologyOtr \) we state that the restriction of a \kl{model} of~\( \ontologyOred \) to \( \conceptRealplus \) satisfies \( \ontologyOhard \).
The translation simply changes all occurrences of \( \conceptA \) (resp. \( \conceptnegA \)) into \( \conceptAplus \) (resp.~\( \conceptAminus \)) and employs \( \conceptRealplus \) to relativise \kl{concepts}.

We start with the definition of \newnotion{\( \ontologyOcorr \)}.
\[
  \ontologyOcorr \deff \lbrace \pconditional{\conceptAplus}{\topconcept}{0.5}{0.5}, \pconditional{\conceptAminus}{\topconcept}{0.5}{0.5}, \pconditional{\conceptAplus}{\conceptAminus}{0}{0} \; \mid \; \conceptA \in \{ \conceptReal \} \cup \ohardconcepts \rbrace
\]
By unfolding the definition of \kl{probabilistic conditionals} we immediately conclude the following facts.
\begin{fact}\label{fact:half}
For any \kl{concept name} \( \conceptA \) we have that \( \interI \mymodels \pconditional{\conceptA}{\topconcept}{0.5}{0.5} \) iff \( |\DeltaI| \) is even and \( |\conceptA^{\interI}| = \frac{1}{2}|\DeltaI| \).
\end{fact} 
\begin{fact}\label{fact:disj}
For any different \kl{concept names} \( \conceptA, \conceptB \) we have that \( \interI \mymodels \pconditional{\conceptA}{\conceptB}{0}{0} \) iff \( \conceptA^{\interI} \cap \conceptB^{\interI} = \emptyset \).
\end{fact} 

Now we focus on the ``translating'' \kl{ontology} \newnotion{\( \ontologyOtr \)}.
\newcommand{\tr}{\mathfrak{tr}}
Let \( \tr \) be a translation function defined by \( \tr(\top) = \conceptRealplus \), \( \tr(\conceptA) = \conceptAplus \dland \conceptRealplus \) and \( \tr(\conceptnegA) = \conceptAminus \dland \conceptRealplus \) for all \kl{concept names} \( \conceptA \in \Clang \) as well as \( \tr(\conceptC \dland \conceptD ) = \tr(\conceptC) \dland \tr(\conceptD) \) and \( \tr(\exists{\roler}.\conceptC) = \conceptRealplus \dland \exists{\roler}.(\tr(\conceptC) \dland \conceptRealplus) \) for complex concepts.
The ontology \( \ontologyOtr \) is obtained by replacing each \kl{GCI} \( \conceptC \dlsubseteq \conceptD \) from \( \ontologyO \) with \( \tr(\conceptC) \dlsubseteq \tr(\conceptD) \).
Finally, we put \newnotion{\( \ontologyOred \)} \( \deff \ontologyOcorr \cup \ontologyOtr \).

Note that the \kl{size} of \( \ontologyOred \) is polynomial in \(|\ontologyOhard| \).
For more intuitions, consult the picture below.

\begin{figure}[H]
  \centering

      
    \tikzset {_h4veaj61h/.code = {\pgfsetadditionalshadetransform{ \pgftransformshift{\pgfpoint{0 bp } { 0 bp }  }  \pgftransformrotate{0 }  \pgftransformscale{2 }  }}}
    \pgfdeclarehorizontalshading{_48nml32ee}{150bp}{rgb(0bp)=(0.82,0.01,0.11);
    rgb(37.5bp)=(0.82,0.01,0.11);
    rgb(49.75bp)=(0.82,0.01,0.11);
    rgb(50.08928575686046bp)=(0.25,0.46,0.02);
    rgb(62.5bp)=(0.25,0.46,0.02);
    rgb(100bp)=(0.25,0.46,0.02)}

      
    \tikzset {_3h06o0nmh/.code = {\pgfsetadditionalshadetransform{ \pgftransformshift{\pgfpoint{0 bp } { 0 bp }  }  \pgftransformrotate{0 }  \pgftransformscale{2 }  }}}
    \pgfdeclarehorizontalshading{_z5yp32z2b}{150bp}{rgb(0bp)=(1,1,1);
    rgb(37.5bp)=(1,1,1);
    rgb(49.75bp)=(1,1,1);
    rgb(50.08928575686046bp)=(0.25,0.46,0.02);
    rgb(62.5bp)=(0.25,0.46,0.02);
    rgb(100bp)=(0.25,0.46,0.02)}

      
    \tikzset {_w58kl4lzb/.code = {\pgfsetadditionalshadetransform{ \pgftransformshift{\pgfpoint{0 bp } { 0 bp }  }  \pgftransformrotate{0 }  \pgftransformscale{2 }  }}}
    \pgfdeclarehorizontalshading{_rkcmni1lf}{150bp}{rgb(0bp)=(1,1,1);
    rgb(37.5bp)=(1,1,1);
    rgb(49.75bp)=(1,1,1);
    rgb(50.08928575686046bp)=(0.25,0.46,0.02);
    rgb(62.5bp)=(0.25,0.46,0.02);
    rgb(100bp)=(0.25,0.46,0.02)}

      
    \tikzset {_2gfllr34p/.code = {\pgfsetadditionalshadetransform{ \pgftransformshift{\pgfpoint{0 bp } { 0 bp }  }  \pgftransformrotate{0 }  \pgftransformscale{2 }  }}}
    \pgfdeclarehorizontalshading{_hctzht9te}{150bp}{rgb(0bp)=(0.82,0.01,0.11);
    rgb(37.5bp)=(0.82,0.01,0.11);
    rgb(49.75bp)=(0.82,0.01,0.11);
    rgb(50.08928575686046bp)=(0.25,0.46,0.02);
    rgb(62.5bp)=(0.25,0.46,0.02);
    rgb(100bp)=(0.25,0.46,0.02)}

      
    \tikzset {_u95ujj9wj/.code = {\pgfsetadditionalshadetransform{ \pgftransformshift{\pgfpoint{0 bp } { 0 bp }  }  \pgftransformrotate{0 }  \pgftransformscale{2 }  }}}
    \pgfdeclarehorizontalshading{_b520fryqd}{150bp}{rgb(0bp)=(0.82,0.01,0.11);
    rgb(37.5bp)=(0.82,0.01,0.11);
    rgb(49.75bp)=(0.82,0.01,0.11);
    rgb(50.17857419592994bp)=(1,1,1);
    rgb(62.5bp)=(1,1,1);
    rgb(100bp)=(1,1,1)}

      
    \tikzset {_o1n4azwiz/.code = {\pgfsetadditionalshadetransform{ \pgftransformshift{\pgfpoint{0 bp } { 0 bp }  }  \pgftransformrotate{0 }  \pgftransformscale{2 }  }}}
    \pgfdeclarehorizontalshading{_p25pouabs}{150bp}{rgb(0bp)=(1,1,1);
    rgb(37.5bp)=(1,1,1);
    rgb(50bp)=(0.95,0.95,0.95);
    rgb(50.338807829788756bp)=(1,1,1);
    rgb(62.5bp)=(0.96,0.96,0.96);
    rgb(100bp)=(0.96,0.96,0.96)}

      
    \tikzset {_guic4zdni/.code = {\pgfsetadditionalshadetransform{ \pgftransformshift{\pgfpoint{0 bp } { 0 bp }  }  \pgftransformrotate{0 }  \pgftransformscale{2 }  }}}
    \pgfdeclarehorizontalshading{_2txb3i5jy}{150bp}{rgb(0bp)=(0.82,0.01,0.11);
    rgb(37.5bp)=(0.82,0.01,0.11);
    rgb(49.75bp)=(0.82,0.01,0.11);
    rgb(50.17857419592994bp)=(1,1,1);
    rgb(62.5bp)=(1,1,1);
    rgb(100bp)=(1,1,1)}

      
    \tikzset {_e9nwg3c5s/.code = {\pgfsetadditionalshadetransform{ \pgftransformshift{\pgfpoint{0 bp } { 0 bp }  }  \pgftransformrotate{0 }  \pgftransformscale{2 }  }}}
    \pgfdeclarehorizontalshading{_jqlkeezhg}{150bp}{rgb(0bp)=(1,1,1);
    rgb(37.5bp)=(1,1,1);
    rgb(49.75bp)=(1,1,1);
    rgb(50.08928575686046bp)=(0.25,0.46,0.02);
    rgb(62.5bp)=(0.25,0.46,0.02);
    rgb(100bp)=(0.25,0.46,0.02)}

      
    \tikzset {_elxfkt69c/.code = {\pgfsetadditionalshadetransform{ \pgftransformshift{\pgfpoint{0 bp } { 0 bp }  }  \pgftransformrotate{0 }  \pgftransformscale{2 }  }}}
    \pgfdeclarehorizontalshading{_b9vinuebn}{150bp}{rgb(0bp)=(0.82,0.01,0.11);
    rgb(37.5bp)=(0.82,0.01,0.11);
    rgb(49.75bp)=(0.82,0.01,0.11);
    rgb(50.17857419592994bp)=(1,1,1);
    rgb(62.5bp)=(1,1,1);
    rgb(100bp)=(1,1,1)}
    \tikzset{every picture/.style={line width=0.75pt}} 

    \begin{tikzpicture}[x=0.75pt,y=0.75pt,yscale=-1,xscale=1]

    \draw   (3.86,83.25) .. controls (3.86,43.29) and (41.02,10.89) .. (86.86,10.89) .. controls (132.7,10.89) and (169.86,43.29) .. (169.86,83.25) .. controls (169.86,123.21) and (132.7,155.61) .. (86.86,155.61) .. controls (41.02,155.61) and (3.86,123.21) .. (3.86,83.25) -- cycle ;
    \path  [shading=_48nml32ee,_h4veaj61h] (43.68,118.91) .. controls (43.68,112.83) and (48.42,107.89) .. (54.27,107.89) .. controls (60.12,107.89) and (64.86,112.83) .. (64.86,118.91) .. controls (64.86,125) and (60.12,129.93) .. (54.27,129.93) .. controls (48.42,129.93) and (43.68,125) .. (43.68,118.91) -- cycle ; 
     \draw   (43.68,118.91) .. controls (43.68,112.83) and (48.42,107.89) .. (54.27,107.89) .. controls (60.12,107.89) and (64.86,112.83) .. (64.86,118.91) .. controls (64.86,125) and (60.12,129.93) .. (54.27,129.93) .. controls (48.42,129.93) and (43.68,125) .. (43.68,118.91) -- cycle ; 

    \draw  [draw opacity=0][fill={rgb, 255:red, 223; green, 111; blue, 111 },fill opacity=1 ] (170.77,44.74) -- (229.8,44.74) -- (229.8,31.01) -- (269.16,58.46) -- (229.8,85.9) -- (229.8,72.18) -- (170.77,72.18) -- (184.49,58.46) -- cycle ;
    \draw [color={rgb, 255:red, 65; green, 117; blue, 5 }  ,draw opacity=1 ]   (59.16,129.43) .. controls (61.43,129.41) and (62.69,130.53) .. (62.95,132.78) .. controls (63.66,135.11) and (65.17,135.9) .. (67.46,135.13) .. controls (69.49,134.12) and (71.05,134.64) .. (72.13,136.7) .. controls (73.54,138.71) and (75.17,139.02) .. (77.04,137.62) .. controls (78.75,136.09) and (80.39,136.16) .. (81.98,137.85) .. controls (83.85,139.4) and (85.56,139.23) .. (87.12,137.33) .. controls (88.2,135.39) and (89.82,134.96) .. (91.98,136.04) .. controls (94.15,136.98) and (95.64,136.33) .. (96.46,134.08) .. controls (97.03,131.81) and (98.49,130.88) .. (100.84,131.27) .. controls (103.12,131.57) and (104.3,130.55) .. (104.38,128.22) .. controls (104.28,125.91) and (105.41,124.66) .. (107.76,124.47) .. controls (110.19,124) and (111.15,122.66) .. (110.63,120.46) .. controls (110.19,117.96) and (111.08,116.41) .. (113.31,115.8) .. controls (115.56,114.98) and (116.29,113.41) .. (115.5,111.09) -- (116.52,108.52) -- (118.84,101.47) ;
    \draw [shift={(119.27,99.93)}, rotate = 464.99] [color={rgb, 255:red, 65; green, 117; blue, 5 }  ,draw opacity=1 ][line width=0.75]    (10.93,-3.29) .. controls (6.95,-1.4) and (3.31,-0.3) .. (0,0) .. controls (3.31,0.3) and (6.95,1.4) .. (10.93,3.29)   ;
    \path  [shading=_z5yp32z2b,_3h06o0nmh] (68.18,63.41) .. controls (68.18,57.33) and (72.92,52.39) .. (78.77,52.39) .. controls (84.62,52.39) and (89.36,57.33) .. (89.36,63.41) .. controls (89.36,69.5) and (84.62,74.43) .. (78.77,74.43) .. controls (72.92,74.43) and (68.18,69.5) .. (68.18,63.41) -- cycle ; 
     \draw   (68.18,63.41) .. controls (68.18,57.33) and (72.92,52.39) .. (78.77,52.39) .. controls (84.62,52.39) and (89.36,57.33) .. (89.36,63.41) .. controls (89.36,69.5) and (84.62,74.43) .. (78.77,74.43) .. controls (72.92,74.43) and (68.18,69.5) .. (68.18,63.41) -- cycle ; 

    \draw [color={rgb, 255:red, 0; green, 0; blue, 0 }  ,draw opacity=1 ]   (108.68,88.91) -- (65.62,111.95) ;
    \draw [shift={(63.86,112.89)}, rotate = 331.85] [color={rgb, 255:red, 0; green, 0; blue, 0 }  ,draw opacity=1 ][line width=0.75]    (10.93,-3.29) .. controls (6.95,-1.4) and (3.31,-0.3) .. (0,0) .. controls (3.31,0.3) and (6.95,1.4) .. (10.93,3.29)   ;
    \draw [color={rgb, 255:red, 0; green, 0; blue, 0 }  ,draw opacity=1 ]   (89.86,66.89) -- (112.59,78.96) ;
    \draw [shift={(114.36,79.89)}, rotate = 207.95] [color={rgb, 255:red, 0; green, 0; blue, 0 }  ,draw opacity=1 ][line width=0.75]    (10.93,-3.29) .. controls (6.95,-1.4) and (3.31,-0.3) .. (0,0) .. controls (3.31,0.3) and (6.95,1.4) .. (10.93,3.29)   ;
    \draw [color={rgb, 255:red, 65; green, 117; blue, 5 }  ,draw opacity=1 ]   (73.36,53.39) .. controls (70.87,52.6) and (69.91,51.15) .. (70.46,49.05) .. controls (70.65,46.58) and (69.57,45.27) .. (67.21,45.11) .. controls (65.01,45.28) and (63.82,44.2) .. (63.64,41.87) .. controls (62.91,39.41) and (61.4,38.57) .. (59.11,39.35) .. controls (57.39,40.6) and (55.77,40.36) .. (54.25,38.63) .. controls (52.29,37.3) and (50.77,37.77) .. (49.69,40.04) .. controls (49.39,42.29) and (48.15,43.34) .. (45.98,43.17) .. controls (43.51,43.76) and (42.61,45.19) .. (43.28,47.47) .. controls (44.29,49.48) and (43.79,51.11) .. (41.76,52.38) .. controls (39.92,53.87) and (39.86,55.54) .. (41.59,57.37) .. controls (43.52,58.54) and (44.02,60.13) .. (43.07,62.14) .. controls (42.83,64.6) and (43.97,65.78) .. (46.49,65.69) .. controls (48.46,64.78) and (50.03,65.31) .. (51.21,67.28) .. controls (52.95,69.03) and (54.6,69.01) .. (56.16,67.22) -- (59.34,66.62) -- (66.78,64.04) ;
    \draw [shift={(68.18,63.41)}, rotate = 515.0799999999999] [color={rgb, 255:red, 65; green, 117; blue, 5 }  ,draw opacity=1 ][line width=0.75]    (10.93,-3.29) .. controls (6.95,-1.4) and (3.31,-0.3) .. (0,0) .. controls (3.31,0.3) and (6.95,1.4) .. (10.93,3.29)   ;
    \draw   (272.39,45.53) .. controls (272.39,27.77) and (286.78,13.37) .. (304.55,13.37) -- (598.1,13.37) .. controls (615.86,13.37) and (630.26,27.77) .. (630.26,45.53) -- (630.26,142.01) .. controls (630.26,159.77) and (615.86,174.17) .. (598.1,174.17) -- (304.55,174.17) .. controls (286.78,174.17) and (272.39,159.77) .. (272.39,142.01) -- cycle ;
    \draw   (275.26,93.15) .. controls (275.26,53.19) and (322.76,20.79) .. (381.36,20.79) .. controls (439.95,20.79) and (487.46,53.19) .. (487.46,93.15) .. controls (487.46,133.11) and (439.95,165.51) .. (381.36,165.51) .. controls (322.76,165.51) and (275.26,133.11) .. (275.26,93.15) -- cycle ;
    \draw [color={rgb, 255:red, 65; green, 117; blue, 5 }  ,draw opacity=1 ]   (330.16,138.93) .. controls (332.43,138.91) and (333.69,140.03) .. (333.95,142.28) .. controls (334.66,144.61) and (336.17,145.4) .. (338.46,144.63) .. controls (340.49,143.62) and (342.05,144.14) .. (343.13,146.2) .. controls (344.54,148.21) and (346.17,148.52) .. (348.04,147.12) .. controls (349.75,145.59) and (351.39,145.66) .. (352.98,147.35) .. controls (354.85,148.9) and (356.56,148.73) .. (358.12,146.83) .. controls (359.2,144.89) and (360.82,144.46) .. (362.98,145.54) .. controls (365.15,146.48) and (366.64,145.83) .. (367.46,143.58) .. controls (368.03,141.31) and (369.49,140.38) .. (371.84,140.77) .. controls (374.12,141.07) and (375.3,140.05) .. (375.38,137.72) .. controls (375.28,135.41) and (376.41,134.16) .. (378.76,133.97) .. controls (381.19,133.5) and (382.15,132.16) .. (381.63,129.96) .. controls (381.19,127.46) and (382.08,125.91) .. (384.31,125.3) .. controls (386.56,124.48) and (387.29,122.91) .. (386.5,120.59) -- (387.52,118.02) -- (389.84,110.97) ;
    \draw [shift={(390.27,109.43)}, rotate = 464.99] [color={rgb, 255:red, 65; green, 117; blue, 5 }  ,draw opacity=1 ][line width=0.75]    (10.93,-3.29) .. controls (6.95,-1.4) and (3.31,-0.3) .. (0,0) .. controls (3.31,0.3) and (6.95,1.4) .. (10.93,3.29)   ;
    \draw [color={rgb, 255:red, 0; green, 0; blue, 0 }  ,draw opacity=1 ]   (379.68,98.41) -- (336.62,121.45) ;
    \draw [shift={(334.86,122.39)}, rotate = 331.85] [color={rgb, 255:red, 0; green, 0; blue, 0 }  ,draw opacity=1 ][line width=0.75]    (10.93,-3.29) .. controls (6.95,-1.4) and (3.31,-0.3) .. (0,0) .. controls (3.31,0.3) and (6.95,1.4) .. (10.93,3.29)   ;
    \draw [color={rgb, 255:red, 0; green, 0; blue, 0 }  ,draw opacity=1 ]   (360.86,76.39) -- (383.59,88.46) ;
    \draw [shift={(385.36,89.39)}, rotate = 207.95] [color={rgb, 255:red, 0; green, 0; blue, 0 }  ,draw opacity=1 ][line width=0.75]    (10.93,-3.29) .. controls (6.95,-1.4) and (3.31,-0.3) .. (0,0) .. controls (3.31,0.3) and (6.95,1.4) .. (10.93,3.29)   ;
    \draw [color={rgb, 255:red, 65; green, 117; blue, 5 }  ,draw opacity=1 ]   (344.36,62.89) .. controls (341.87,62.1) and (340.91,60.65) .. (341.46,58.55) .. controls (341.65,56.08) and (340.57,54.77) .. (338.21,54.61) .. controls (336.01,54.78) and (334.82,53.7) .. (334.64,51.37) .. controls (333.91,48.91) and (332.4,48.07) .. (330.11,48.85) .. controls (328.39,50.1) and (326.77,49.86) .. (325.25,48.13) .. controls (323.29,46.8) and (321.77,47.27) .. (320.69,49.54) .. controls (320.39,51.79) and (319.15,52.84) .. (316.98,52.67) .. controls (314.51,53.26) and (313.61,54.69) .. (314.28,56.97) .. controls (315.29,58.98) and (314.79,60.61) .. (312.76,61.88) .. controls (310.92,63.37) and (310.86,65.04) .. (312.59,66.87) .. controls (314.52,68.04) and (315.02,69.63) .. (314.07,71.64) .. controls (313.83,74.1) and (314.97,75.28) .. (317.49,75.19) .. controls (319.46,74.28) and (321.03,74.81) .. (322.21,76.78) .. controls (323.95,78.53) and (325.6,78.51) .. (327.16,76.72) -- (330.34,76.12) -- (337.78,73.54) ;
    \draw [shift={(339.18,72.91)}, rotate = 515.0799999999999] [color={rgb, 255:red, 65; green, 117; blue, 5 }  ,draw opacity=1 ][line width=0.75]    (10.93,-3.29) .. controls (6.95,-1.4) and (3.31,-0.3) .. (0,0) .. controls (3.31,0.3) and (6.95,1.4) .. (10.93,3.29)   ;
    \path  [shading=_rkcmni1lf,_w58kl4lzb] (339.18,72.91) .. controls (339.18,66.83) and (343.92,61.89) .. (349.77,61.89) .. controls (355.62,61.89) and (360.36,66.83) .. (360.36,72.91) .. controls (360.36,79) and (355.62,83.93) .. (349.77,83.93) .. controls (343.92,83.93) and (339.18,79) .. (339.18,72.91) -- cycle ; 
     \draw   (339.18,72.91) .. controls (339.18,66.83) and (343.92,61.89) .. (349.77,61.89) .. controls (355.62,61.89) and (360.36,66.83) .. (360.36,72.91) .. controls (360.36,79) and (355.62,83.93) .. (349.77,83.93) .. controls (343.92,83.93) and (339.18,79) .. (339.18,72.91) -- cycle ; 

    \path  [shading=_hctzht9te,_2gfllr34p] (319.57,127.91) .. controls (319.57,121.83) and (324.31,116.89) .. (330.16,116.89) .. controls (336.01,116.89) and (340.75,121.83) .. (340.75,127.91) .. controls (340.75,134) and (336.01,138.93) .. (330.16,138.93) .. controls (324.31,138.93) and (319.57,134) .. (319.57,127.91) -- cycle ; 
     \draw   (319.57,127.91) .. controls (319.57,121.83) and (324.31,116.89) .. (330.16,116.89) .. controls (336.01,116.89) and (340.75,121.83) .. (340.75,127.91) .. controls (340.75,134) and (336.01,138.93) .. (330.16,138.93) .. controls (324.31,138.93) and (319.57,134) .. (319.57,127.91) -- cycle ; 

    \path  [shading=_b520fryqd,_u95ujj9wj] (506.28,49.91) .. controls (506.28,43.83) and (511.02,38.89) .. (516.87,38.89) .. controls (522.72,38.89) and (527.46,43.83) .. (527.46,49.91) .. controls (527.46,56) and (522.72,60.93) .. (516.87,60.93) .. controls (511.02,60.93) and (506.28,56) .. (506.28,49.91) -- cycle ; 
     \draw   (506.28,49.91) .. controls (506.28,43.83) and (511.02,38.89) .. (516.87,38.89) .. controls (522.72,38.89) and (527.46,43.83) .. (527.46,49.91) .. controls (527.46,56) and (522.72,60.93) .. (516.87,60.93) .. controls (511.02,60.93) and (506.28,56) .. (506.28,49.91) -- cycle ; 

    \path  [shading=_p25pouabs,_o1n4azwiz] (504.77,139.91) .. controls (504.77,133.83) and (509.51,128.89) .. (515.36,128.89) .. controls (521.21,128.89) and (525.95,133.83) .. (525.95,139.91) .. controls (525.95,146) and (521.21,150.93) .. (515.36,150.93) .. controls (509.51,150.93) and (504.77,146) .. (504.77,139.91) -- cycle ; 
     \draw   (504.77,139.91) .. controls (504.77,133.83) and (509.51,128.89) .. (515.36,128.89) .. controls (521.21,128.89) and (525.95,133.83) .. (525.95,139.91) .. controls (525.95,146) and (521.21,150.93) .. (515.36,150.93) .. controls (509.51,150.93) and (504.77,146) .. (504.77,139.91) -- cycle ; 

    \path  [shading=_2txb3i5jy,_guic4zdni] (108.68,88.91) .. controls (108.68,82.83) and (113.42,77.89) .. (119.27,77.89) .. controls (125.12,77.89) and (129.86,82.83) .. (129.86,88.91) .. controls (129.86,95) and (125.12,99.93) .. (119.27,99.93) .. controls (113.42,99.93) and (108.68,95) .. (108.68,88.91) -- cycle ; 
     \draw   (108.68,88.91) .. controls (108.68,82.83) and (113.42,77.89) .. (119.27,77.89) .. controls (125.12,77.89) and (129.86,82.83) .. (129.86,88.91) .. controls (129.86,95) and (125.12,99.93) .. (119.27,99.93) .. controls (113.42,99.93) and (108.68,95) .. (108.68,88.91) -- cycle ; 

    \path  [shading=_jqlkeezhg,_e9nwg3c5s] (506.38,95.71) .. controls (506.38,89.63) and (511.12,84.69) .. (516.97,84.69) .. controls (522.82,84.69) and (527.56,89.63) .. (527.56,95.71) .. controls (527.56,101.8) and (522.82,106.73) .. (516.97,106.73) .. controls (511.12,106.73) and (506.38,101.8) .. (506.38,95.71) -- cycle ; 
     \draw   (506.38,95.71) .. controls (506.38,89.63) and (511.12,84.69) .. (516.97,84.69) .. controls (522.82,84.69) and (527.56,89.63) .. (527.56,95.71) .. controls (527.56,101.8) and (522.82,106.73) .. (516.97,106.73) .. controls (511.12,106.73) and (506.38,101.8) .. (506.38,95.71) -- cycle ; 

    \path  [shading=_b9vinuebn,_elxfkt69c] (379.68,98.41) .. controls (379.68,92.33) and (384.42,87.39) .. (390.27,87.39) .. controls (396.12,87.39) and (400.86,92.33) .. (400.86,98.41) .. controls (400.86,104.5) and (396.12,109.43) .. (390.27,109.43) .. controls (384.42,109.43) and (379.68,104.5) .. (379.68,98.41) -- cycle ; 
     \draw   (379.68,98.41) .. controls (379.68,92.33) and (384.42,87.39) .. (390.27,87.39) .. controls (396.12,87.39) and (400.86,92.33) .. (400.86,98.41) .. controls (400.86,104.5) and (396.12,109.43) .. (390.27,109.43) .. controls (384.42,109.43) and (379.68,104.5) .. (379.68,98.41) -- cycle ; 

    \draw  [draw opacity=0][fill={rgb, 255:red, 223; green, 111; blue, 111 }, fill opacity=1 ] (262.16,139.18) -- (203.13,139.18) -- (203.13,152.9) -- (163.77,125.46) -- (203.13,98.01) -- (203.13,111.74) -- (262.16,111.74) -- (248.44,125.46) -- cycle ;

        \draw (194.47,78.47) node [anchor=north west][inner sep=0.75pt]  [font=\scriptsize,color={rgb, 255:red, 255; green, 255; blue, 255 }  ,opacity=1 ] [align=left] {reduction};
        \draw (113.9,58) node [anchor=north west][inner sep=0.75pt]    {$\domelemd_2{:} \; \conceptA, \conceptnegB$};
        \draw (22,89) node [anchor=north west][inner sep=0.75pt]    {$\domelemd_3{:} \; \conceptA, \conceptB$};
        \draw (74,32) node [anchor=north west][inner sep=0.75pt]    {$\domelemd_1{:} \; \conceptnegA, \conceptB$};
        \draw (530.01,39) node [anchor=north west][inner sep=0.75pt]    {$\domelemd_1'{:} \; \conceptAplus, \conceptBminus$};
        \draw (531.73,86) node [anchor=north west][inner sep=0.75pt]    {$\domelemd_2'{:} \; \conceptAminus, \conceptBplus$};
        \draw (531.66,131) node [anchor=north west][inner sep=0.75pt]    {$\domelemd_3'{:} \; \conceptAminus, \conceptBminus$};
        \draw (285,99) node [anchor=north west][inner sep=0.75pt]    {$\domelemd_3{:} \; \conceptAplus, \conceptBplus$};
        \draw (402.9,91) node [anchor=north west][inner sep=0.75pt]    {$\domelemd_2{:} \; \conceptAplus, \conceptBminus$};
        \draw (362,59) node [anchor=north west][inner sep=0.75pt]    {$\domelemd_1{:} \; \conceptAminus, \conceptBplus$};
        \draw (365,30) node [anchor=north west][inner sep=0.75pt]    {$\conceptRealplus$};
        \draw (455,20) node [anchor=north west][inner sep=0.75pt]    {$\conceptRealminus$};

        \draw (215,152) node [anchor=north west][inner sep=0.75pt]    {$\interJ \mymodels \ontologyOred$};
        \draw (145,12) node [anchor=north west][inner sep=0.75pt]    {$\interI \mymodels \ontologyOhard$};

        \draw (190,53.47) node [anchor=north west][inner sep=0.75pt]  [font=\scriptsize,color={rgb, 255:red, 255; green, 255; blue, 255 }  ,opacity=1 ] [align=left] {\cref{lemma:from_el_to_sel}};
        \draw (190,119.47) node [anchor=north west][inner sep=0.75pt]  [font=\scriptsize,color={rgb, 255:red, 255; green, 255; blue, 255 }  ,opacity=1 ] [align=left] {\cref{lemma:from_sel_to_el}};

    \end{tikzpicture}
\end{figure}

\subsection{Correctness of the reduction}
Let us start with an auxiliary notion of interpretations that are \newnotion{good-for-encoding}.
We say that $\interJ$ is \kl{good-for-encoding} if for all \kl{concept} names $\conceptA \in \ohardconcepts$ it satisfies $\conceptA^{\interJ} = \conceptAplus^{\interJ}$ and $\conceptAminus^{\interJ} = \DeltaJ \setminus \conceptA^{\interJ}$.
The following lemma relates the translation function $\tr$, \kl{good-for-encoding} interpretations and their submodels.

\begin{lemma}[Agreement lemma]\label{lemma:agreement}
Let $\interJ$ be \kl{good-for-encoding} and let $\interI$ be its induced subinterpretation with domain $\conceptRealplus^{\interJ}$.
Then all $\ELneg$ \kl{concepts} $\conceptC$ employing \emph{only} concept names from $\ohardconcepts$ satisfy $\conceptC^{\interI} = \tr(\conceptC)^{\interJ}$.
Moreover for such concepts $\conceptC, \conceptD$ we have: $\interJ \models \tr(\conceptC) \dlsubseteq \tr(\conceptD)$ iff $\interI \models \conceptC \dlsubseteq \conceptD$. 
\end{lemma}
\begin{proof}
We proceed by inductively on the shape of \kl{concepts} $\conceptC$.
The cases for $\conceptC = \top, \conceptA$ or $\conceptnegA$ for $\conceptA \in \Clang$ follow immediately from the definition of $\tr$ and the assumptions $\conceptA^{\interJ} = \conceptAplus^{\interJ}$ and $\conceptAminus^{\interJ} = \DeltaJ \setminus \conceptA^{\interJ}$. 
The case of $\conceptC = \conceptD \dland \conceptE$ follows from the fact that $\tr$ is homomorphic for $\dland$.
Hence, the only interesting case is when $\conceptC = \exists{\roler}.\conceptD$. 
Assuming $\conceptD^{\interI} = \tr(\conceptD)^{\interJ}$ we will show two inclusions. 

\begin{itemize}
  \item For the first inclusion, take $\domelemd \in (\exists{\roler}.\conceptD)^{\interI}$. Thus $\domelemd \in \DeltaI (= \conceptRealplus^{\interJ})$ and there exists an $\domeleme \in \DeltaI$ satisfying both $(\domelemd, \domeleme) \in \roler^{\interI}$ and $\domeleme \in \conceptD^{\interI}$. 
  Note that $\domeleme \in \DeltaI$ implies $\domeleme \in \conceptRealplus^{\interJ}$.
  Moreover, by the equality $\conceptD^{\interI} = \tr(\conceptD)^{\interJ}$ we have $\domeleme \in (\conceptRealplus \dland \tr(\conceptD))^{\interJ}$.
  Since $\roler^{\interI} \subseteq \roler^{\interJ}$ we infer $\domelemd \in (\exists{\roler}.(\conceptRealplus \dland \tr(\conceptD)))^{\interJ}$, but because $\domelemd \in \conceptRealplus^{\interJ}$ we can conclude that $\domelemd \in (\conceptRealplus \dland \exists{\roler}.(\conceptRealplus \dland \tr(\conceptD)))^{\interJ} = \tr(\exists{\roler}.\conceptD)^{\interJ}$.
  \item For the opposite direction take $\domelemd \in \tr(\exists{\roler}.\conceptD)^{\interJ} = (\conceptRealplus \dland \exists{\roler}.(\conceptRealplus \dland \tr(\conceptD)))^{\interJ}$.
  It implies that $\domelemd \in \conceptRealplus^{\interJ} (= \DeltaI)$ as well as that there is an $\domeleme \in \conceptRealplus^{\interJ} (= \DeltaI)$ witnessing $(\domelemd, \domeleme) \in \roler^{\interJ}$ and $\domeleme \in \tr(\conceptD)^{\interJ} (= \conceptD^{\interI})$. Since both $\domelemd, \domeleme$ belong to $\DeltaI$ we infer that $(\domelemd, \domeleme) \in \roler^{\interI}$ and hence $\domelemd \in (\exists{\roler}.\conceptD)^{\interI}$.
\end{itemize}

For the last statement of the lemma: to show that $\interJ \models \tr(\conceptC) \dlsubseteq \tr(\conceptD)$ iff $\interI \models \conceptC \dlsubseteq \conceptD$ hold, it suffices to invoke $\conceptC^{\interI} = \tr(\conceptC)^{\interJ}$ and $\conceptD^{\interI} = \tr(\conceptD)^{\interJ}$ to see that the inclusions $\conceptC^{\interI} \subseteq \conceptD^{\interI}$ and $\tr(\conceptC)^{\interJ} \subseteq \tr(\conceptD)^{\interJ}$ are equivalent.
\end{proof}

The agreement lemma can now be used to show that the \kl{consistency} of \( \ontologyOred \) implies the  \kl{consistency} of \( \ontologyOhard \).

\begin{lemma}\label{lemma:from_sel_to_el}
If \( \ontologyOred \) is \kl{consistent} then so is \( \ontologyOhard \).
\end{lemma}
\begin{proof}
  Let \( \interJ \) be a model of \( \ontologyOred \) with $\conceptA^{\interJ} \deff \conceptAplus^{\interJ}$ (since $\conceptA$ does not appear in $\ontologyOred$ this can be assumed w.l.o.g.) for all concept names $\conceptA$ from $\ohardconcepts$. 
  By the satisfaction of $\ontologyOcorr$ we know that $\conceptAplus^{\interJ}$ and $\conceptAminus^{\interJ}$ are disjoint and thus $\interJ$ is \kl{good-for-encoding}.
  Hence, take $\interI$ to be its induced subinterpretation with domain $\conceptRealplus^{\interJ}$.
  By applying~\cref{lemma:agreement} we know that for each \kl{GCI} $\conceptC \dlsubseteq \conceptD$ from $\ontologyOhard$ the satisfaction of $\interJ \models \tr(\conceptC) \dlsubseteq \tr(\conceptD)$ implies $\interI \models \conceptC \dlsubseteq \conceptD$.
  Thus we get~$\interI \models \ontologyOhard$, which implies that $\ontologyOhard$ is \kl{consistent}.
\end{proof}

We next show that the \kl{consistency} of \( \ontologyOhard \) implies the \kl{consistency} of \( \ontologyOred \). 
In the proof we basically take a \kl{model} of \( \ontologyOhard \), duplicate each \kl{domain} element and define the memberships of fresh \kl{concepts} introduced by \( \ontologyOred \).
Such \kl{concepts} are defined in such a way that if an element from a \kl{model} \( \interI \) of \( \ontologyOhard \) is a member of \( \conceptA^{\interI} \) then the corresponding element in a constructed \kl{model} \( \interJ \) for~\( \ontologyOred \) is a member of \( \conceptAplus^{\interJ} \) while its copy belongs to \( \conceptAminus^{\interJ} \).
In this way, the total number of elements in every \kl{concept} is always equal to half of the \kl{domain}.

\begin{lemma}\label{lemma:from_el_to_sel}
If \( \ontologyOhard \) is \kl{consistent} then so is \( \ontologyOred \).
\end{lemma}
\begin{proof}
  Let \( \interI \mymodels \ontologyOhard \) and let \( \DeltaI = \{ \domelemd_1, \domelemd_2, \ldots, \domelemd_n \} \).
  We define an \kl{interpretation} \( \interJ \) as follows:
  \begin{enumerate}
    \item\label{eq:proof1:domain}  \( \DeltaJ \deff \{ \domelemd_1, \domelemd_1', \domelemd_2, \domelemd_2', \ldots, \domelemd_n, \domelemd_n' \}. \)
    \item\label{eq:proof1:concepts} For all \kl{concept names} \( \conceptA \in  \ohardconcepts  \) we put 
      \begin{itemize}
        \item \( \conceptAplus^{\interJ} \deff \conceptA^{\interJ} \deff \{ \domelemd_i \; \mid \; \domelemd_i \in \conceptA^{\interI} \} \cup \{ \domelemd_i' \; \mid \; \domelemd_i \not\in \conceptA^{\interI}  \} \) and \( \conceptAminus^{\interJ} \deff \{ \domelemd_i \; \mid \; \domelemd_i \not\in \conceptA^{\interI} \} \cup \{ \domelemd_i' \; \mid \; \domelemd_i \in \conceptA^{\interI}  \} \),
        \item \( \conceptRealplus^{\interJ} \deff \DeltaI \) and \( \conceptRealminus^{\interJ} \deff \DeltaJ \setminus \DeltaI, \)
      \end{itemize}
      and for all other \kl{concept names} \( \conceptB \) we put \( \conceptB^{\interI} \deff \DeltaI \).
    \item\label{eq:proof1:roles} For each \kl{role name} \( \roler \) we put \( \roler^{\interJ} \deff \roler^{\interI} \). 
  \end{enumerate}
  We first show \( \interJ \mymodels \ontologyOcorr \). To this end, take any name \( \conceptA \in \{ \conceptReal \} \cup \ohardconcepts \).
  We prove~\( \interJ \mymodels \pconditional{\conceptAplus}{\conceptAminus}{0}{0} \), which by \cref{fact:disj} is equivalent to showing disjointness of \( \conceptAplus^{\interJ} \) and~\( \conceptAminus^{\interJ} \). 
  Assume the contrary, \ie{} that there is a \kl{domain} element \( \domelemd \in \conceptAplus^{\interJ} \cap \conceptAminus^{\interJ} \). 
  If \( \domelemd = \domelemd_i \) for some index \( i \) then, by \cref{eq:proof1:concepts}, it means that~\( \domelemd_i \in \conceptA^{\interI} \) and \( \domelemd_i \not\in \conceptA^{\interI} \) at the same time, which is clearly not possible.
  The case when \( \domelemd = \domelemd_i' \) for some index \(i\) is treated similarly.
  Next, by invoking~\cref{eq:proof1:concepts} of the definition of \( \interJ \) we can perform some basic calculations:
  \vspace{-0.75em}
  \begin{multline*}
  |\conceptAplus^{\interJ}| = |\{ \domelemd_i \; \mid \; \domelemd_i \in \conceptA^{\interI} \}| + |\{ \domelemd_i' \; \mid \; \domelemd_i \not\in \conceptA^{\interI} \}|  = |\{ \domelemd_i \; \mid \; \domelemd_i \in \conceptA^{\interI} \}| + |\{ \domelemd_i \; \mid \; \domelemd_i \not\in \conceptA^{\interI} \}|   = |\DeltaI| =  0.5 \cdot | \DeltaJ |.\\
  \quad |\conceptAminus^{\interJ}| = |\{ \domelemd_i \; \mid \; \domelemd_i \not\in \conceptA^{\interI} \}| + |\{ \domelemd_i' \; \mid \; \domelemd_i \in \conceptA^{\interI} \}|  = |\{ \domelemd_i' \; \mid \; \domelemd_i \not\in \conceptA^{\interI} \}| + |\{ \domelemd_i' \; \mid \; \domelemd_i \in \conceptA^{\interI} \}| = |\{ \domelemd_1', \ldots, \domelemd_n' \}| =  0.5 \cdot | \DeltaJ |.
  \end{multline*}
  Hence by \cref{fact:half} we get \( \interJ \mymodels \pconditional{\conceptAplus}{\topconcept}{0.5}{0.5} \)
  and \( \interJ \mymodels \pconditional{\conceptAminus}{\topconcept}{0.5}{0.5} \), finishing the proof of  \( \interJ \mymodels \ontologyOcorr \).

  To prove \( \interJ \mymodels \ontologyOtr \) we take any \kl{GCI} $\tr(\conceptC) \dlsubseteq \tr(\conceptD)$ from $\ontologyOtr$.
  From $\interI \models \ontologyOhard$, we know that $\interI \models \conceptC \dlsubseteq \conceptD$. 
  Note that $\interJ$ is \kl{good-for-encoding}, hence by~\cref{lemma:agreement} it follows that $\interJ \models \tr(\conceptC) \dlsubseteq \tr(\conceptD)$. 
  Thus $\interJ \models \ontologyOtr$.
\end{proof}

\cref{lemma:from_el_to_sel} and \cref{lemma:from_sel_to_el} show that the presented reduction is correct. 
Since our reduction is polynomial for every $\ELneg{}$ ontology $\ontologyOhard$, from~\cref{prop:exp-hard} we can conclude the main theorem of this note.
\begin{thm}
The \kl{consistency problem} for Statistical \( \EL \) is \( \ExpTime \)-hard and remains \( \ExpTime \)-hard even if the only numbers used in \kl{probabilistic conditionals} are \( 0, 0.5 \) and \( 1 \).
\end{thm}
With the already mentioned $\ExpTime$ upper bound we conclude~\cref{thm:main}.

\section{Conclusions}\label{sec:conclusions}
  We have proved that the \kl{consistency problem} for Statistical \( \EL \) \kl{ontologies} is \( \ExpTime \)-complete.
  While the upper bound was derived from the works of Baader and Ecke~\cite{BaaderE17} or Lutz and Schr{\"{o}}der on (the extensions of) \( \ALC \) with cardinality constraints, the main contribution of the paper is the lower bound. 
  Our proof went via a reduction from the \kl{consistency problem} for \( \ELneg \) \kl{ontologies} and heavily relied on the fact that \kl{probabilistic conditionals} can express that exactly half of the \kl{domain} elements belong to a certain \kl{concept}.

  An interesting direction for future work is to consider extensions of other well-known decidable fragments of first-order logic with \kl{probabilistic conditionals}. 
  Promising candidates are the guarded fragment~\cite{AndrekaNB98}, the guarded negation fragment~\cite{BaranyCS15}, the two-variable logic~\cite{GradelKV97} and tamed fragments of existential rules~\cite{MugnierT14}. 
  Another idea is to study query answering~\cite{OrtizS12} in the presence of \kl{probabilistic conditionals}. 
  Some initial results were obtained in~\cite{BaaderBR20}.

\section*{Acknowledgements}
  This work was supported by the Polish Ministry of Science and Higher Education program ``Diamentowy Grant'' no.~DI2017~006447. 
  The author would like to thank Ania Karykowska and Elisa B{\"o}hl for reading the draft of this note as well as Piotr Witkowski for serving as a rubber duck while ``rubber duck debugging'' the proof ideas.
  He also grateful for many grammar corrections and improvement ideas given by anonymous IPL reviewers.

\bibliographystyle{alpha}
\bibliography{references}

\end{document}